\documentclass[a4paper,reqno]{amsart} %Ultima version
\usepackage[all]{xy} %For commutative diagrams
\usepackage{amssymb} %For double arrows
\usepackage{hyperref}
\usepackage{eucal}

\numberwithin{equation}{section}

\usepackage[active]{srcltx} %para localizar lo visualizado

 1  %Caracteres "doble palo".
\newfont{\msi}{msbm8 scaled \magstephalf}  %Caract. doble palo sub\'{\i}ndice
\newfont{\msii}{msbm6 scaled \magstephalf}  %Caract. doble palo sub\'{\i}ndic

\newtheorem{definition}{Definition}[section]

\newtheorem{theorem}[definition]{Theorem}
\newtheorem{proposition}[definition]{Proposition}
\newtheorem{corollary}[definition]{Corollary}
\newtheorem{remarkth}[definition]{Remark}
\newtheorem{example}[definition]{Example}
\newenvironment{remark}{\begin{remarkth}\upshape}{\hfill$\diamond$\end{remarkth}}

\renewcommand{\emph}[1]{{\bfseries\itshape{#1}}}
\newcommand{\comp}{\makebox[7pt]{\raisebox{1.5pt}{\tiny $\circ$}}}

\newcommand{\R}{\mathbb{R}} %Numeros reales
\newcommand{\C}{\mathcal{C}}
\def\alg{(E,\lcf \cdot,\cdot \rcf ,\rho )}

\def\vb{\tau:E\to Q}

\newcommand{\lcf}{\lbrack\! \lbrack}
\newcommand{\rcf}{\rbrack\! \rbrack}

\newcommand{\FL}{\mathbb{F}L}
\makeatletter
\newcommand\prol{\@ifstar{\@proldf}{\@prolpf}} %% if * dual else primal
\def\@prolpf{\@ifnextchar[{\@prolpf@wrt}{\@prolpf@}}
\def\@prolpf@wrt[#1]#2{\@ifnextchar[{\@prolpf@wrt@at{#1}{#2}}{\@prolpf@wrt@{#1}{#2}}}

\def\@prolpf@wrt@at#1#2[#3]{\prolsymbol^{#1}_{#3}#2}
\def\@prolpf@wrt@#1#2{\prolsymbol^{#1}#2}
\def\@prolpf@#1{\@ifnextchar[{\@prolpf@at{#1}}{\@prolpf@@{#1}}}
\def\@prolpf@at#1[#2]{\prolsymbol_{#2}#1}
\def\@prolpf@@#1{\prolsymbol#1}
\def\@proldf{\@ifnextchar[{\@proldf@wrt}{\@proldf@}}
\def\@proldf@wrt[#1]#2{\@ifnextchar[{\@proldf@wrt@at{#1}{#2}}{\@proldf@wrt@{#1}{#2}}}

\def\@proldf@wrt@at#1#2[#3]{\prolsymbol^{*#1}_{#3}#2}
\def\@proldf@wrt@#1#2{\prolsymbol^{*#1}#2}
\def\@proldf@#1{\@ifnextchar[{\@proldf@at{#1}}{\@proldf@@{#1}}}
\def\@proldf@at#1[#2]{\prolsymbol^*_{#2}#1}
\def\@proldf@@#1{\prolsymbol^*#1}
\def\prolsymbol{\mathcal{L}}
\makeatother

\newcommand{\X}{\mathcal{X}}
\newcommand{\Y}{\mathcal{Y}}
\newcommand{\V}{\mathcal{V}}
\newcommand{\T}{\mathcal{T}}
\renewcommand{\P}{\mathcal{P}}

\newcommand{\Ker}{\mathrm{ker\;}}

\newcommand{\D}{\mathcal{D}} %Inversa de \C

\def\lcf{\lbrack\! \lbrack}
\def\rcf{\rbrack\! \rbrack}
\newcommand{\U}{{\mathcal{U}}}
\newcommand{\pr}{\mathrm{pr}}
\begin{document}

\title[]
{Dirac structures and Hamilton--Jacobi theory for Lagrangian
Mechanics on Lie algebroids}

\author[M.\ Leok]{M. Leok}
\address{M.\ Leok: Department of Mathematics, University of California, San Diego, 9500 Gilman Drive, La Jolla, California, USA}
\email{mleok@math.ucsd.edu}

\author[D.\ Sosa]{D.\ Sosa}
\address{D.\ Sosa:
Departamento de Econom\'{\i}a Aplicada y Unidad Asociada ULL-CSIC
Geo\-me\-tr\'{\i}a Diferencial y Mec\'anica Geom\'etrica, Facultad de
CC. EE. y Empresariales, Universidad de La Laguna, La Laguna,
Tenerife, Canary Islands, Spain} \email{dnsosa@ull.es}



\begin{abstract}
This paper develops the notion of implicit Lagrangian systems on
Lie algebroids and a Hamilton--Jacobi theory for this type of
system. The Lie algebroid framework provides a natural
generalization of classical tangent bundle geometry. We define the
notion of an implicit Lagrangian system on a Lie algebroid $E$
using Dirac structures on the Lie algebroid prolongation
$\T^EE^*$. This setting includes degenerate Lagrangian systems
with nonholonomic constraints on Lie algebroids.
\end{abstract}

\maketitle

\tableofcontents
\section{Introduction}

There is a vast literature on the Lagrangian formalism in
mechanics, which is due to the central role played by these
systems in the foundations of modern mathe\-ma\-tics and physics.
In many interesting systems, problems often arise  due to their
singular nature, which gives rise to constraints that address the
fact that the evolution problem is not well-posed ({\sl internal
constraints}). Constraints can also manifest a priori restrictions
on the states of the system which arise due to physical arguments
or from external conditions ({\sl external constraints}). Both
cases are of considerable importance.

Systems with {\sl internal constraints} are quite interesting
since many dynamical systems are given in terms of presymplectic
forms instead of the more habitual symplectic ones. The more
frequent case appears in the Lagrangian formalism of singular
mechanical systems which are commonplace in many physical theories
(as in Yang-Mills theories, gravitation, etc).

Systems subjected to {\sl external constraints} (holonomic and
nonholonomic) have a wide range of applications in many different
areas: engineering, optimal control theory, mathematical economics
(growth economic theory), subriemannian geometry, motion of
microorganisms, etc. Interconnected and implicit systems play a
key role in, for example, controlled mechanical systems like
robots. An important class of implicit mechanical systems is those
with nonholonomic constraints, which has a long and rich history
(see, for instance, \cite{Bl} and \cite{NF}). The Lagrangian and
Hamiltonian approaches for such systems have been extensively
developed (see \cite{KoMa1,VSMa,VF,We}), including symmetry and
reduction (see \cite{BaSn,BlKMM,Koi,KoMa2,MR}).

Some authors have given descriptions of L-C circuits and
nonholonomics systems in the context of Poisson structures (see
\cite{MVSB,VSMa}) and later in the general context of Dirac
structures (see \cite{BlCr,VSMa2}) from a Hamiltonian point of
view. Inspired by these works, Yoshimura and Marsden in
\cite{YM1,YM2} have developed a Lagrangian formalism making use of
the framework of Dirac structures.

Recent investigations have lead to a unifying geometric framework
covering a plethora of particular situations. It is precisely the
underlying structure of a Lie algebroid on the phase space which
allows a unified treatment. This idea was first introduced by
Weinstein \cite{weinstein} in order to define a Lagrangian
formalism which is general enough to account for different types
of systems. The geometry and dynamics on Lie algebroids have been
extensively studied during the past years. In particular, in
\cite{mart}, E. Mart{\'\i}nez developed a geometric formalism of
mechanics on Lie algebroids similar to Klein's formalism of
ordinary Lagrangian mechanics and, more recently, a description of
the Hamiltonian dynamics on a Lie algebroid was given
in~\cite{LMM,Medina}. The key concept in this theory is the
prolongation, ${\mathcal T }^E E$, of the Lie algebroid over the
fiber projection $\tau$ (for the Lagrangian formalism) and the
prolongation, ${\mathcal T}^E E^*$, over the dual fiber projection
$\tau^*: E^*\to Q$ (for the Hamiltonian formalism). See \cite{LMM}
for more details. Of course, when the Lie algebroid is $E=TQ$ we
obtain that ${\mathcal T}^E E=T(TQ)$ and ${\mathcal
T}^EE^*=T(T^*Q)$, recovering the classical case. Another approach
to the theory was discussed in \cite{GGU}.

The notion of nonholonomic systems on a Lie algebroid was
introduced in \cite{CoMa} when studying mechanical control systems
and an approach to mechanical systems on Lie algebroids subject to
linear constraints was presented in \cite{MeLa}. A recent
comprehensive treatment of nonholonomic systems on a Lie algebroid
has been develop in \cite{CoLeMaMa}, where the authors identify
suitable conditions guaranteeing that the system admits a unique
solution and show that many of the properties that standard
nonholonomic systems enjoy have counterparts in the Lie algebroid
setting.

On the other hand, singular or degenerate Lagrangian systems and
vakonomic mechanics on Lie algebroids (obtained through the
application of a constrained variational principle) also have been
studied. In \cite{IMMS}, the authors introduce a constraint
algorithm for presymplectic Lie algebroids which generalizes the
well-known Gotay-Nester-Hinds algorithm (see \cite{GNH}) and
applies it to singular Lagrangian systems on Lie algebroids.
Moreover, they develop a geometric description of vakonomic
mechanics on Lie algebroids using again the constraint algorithm.

As a consequence of all these investigations, one deduces that
there are several reasons for discussing unconstrained
(constrained) Mechanics on Lie algebroids:

i) The inclusive nature of the Lie algebroid framework. In fact,
under the same umbrella, one can consider standard unconstrained
(constrained) mechanical systems, (nonholonomic and vakonomic)
Lagrangian systems on Lie algebras, unconstrained (constrained)
systems evolving on semidirect products or (nonholonomic and
vakonomic) Lagrangian systems with symmetries.

ii) The reduction of a (nonholonomic or vakonomic) mechanical
system on a Lie algebroid is a (nonholonomic or vakonomic)
mechanical system on a Lie algebroid. However, the reduction of a
standard unconstrained (constrained) system on the tangent
(cotangent) bundle of the configuration manifold is not, in
general, a standard unconstrained (constrained) system.

iii) The theory of Lie algebroids gives a natural interpretation
of the use of quasi-coordinates (velocities) in Mechanics
(particularly, in nonholonomic and vakonomic mechanics).

On the other hand, Hamilton--Jacobi theory has been studied for
different type of systems for many years. For degenerate
Lagrangian systems, some work have been done on extending
Hamilton--Jacobi theory, using Dirac's theory of constraints (see,
e.g., \cite{HT}) and from a geometric point of view (see
\cite{CGMMMR}). For nonholonomic systems, in \cite{ILM},
Iglesias-Ponte, de Le\'{o}n and Mart\'{\i}n de Diego generalized the
geometric Hamilton--Jacobi theorem (see Theorem 5.2.4. in
\cite{AM}) to nonholonomic systems, which has been studied further
(see \cite{CGMMMR2,OB,OFBZ}). More recently, in \cite{LOS}, the
authors have presented a Hamilton--Jacobi theory which can deal
with both degeneracy and nonholonomic constraints. In the context
of Lie algebroids, de Le\'{o}n, Marrero and Mart\'{\i}n de Diego have
developed a more general formalism which is also valid for for
nonholonomic systems on a Lie algebroid (see \cite{LMM2}), and, in
\cite{BMMP}, the authors have presented a Hamilton--Jacobi
equation for a Hamiltonian system on a skew-symmetric algebroid.

The goal of this paper is to generalize Hamilton--Jacobi theory to
implicit Lagrangian systems on a Lie algebroid based on Dirac
structures. We introduce the notion of an implicit Lagrangian
system on a Lie algebroid $E$ using the induced generalized Dirac
structure $\D_\U$ on the Lie algebroid prolongation $\T^EE^*$ that
is naturally induced by a vector subbundle $\U$ of $E$ and we
obtain the Hamilton--Jacobi theorem for this kind of systems. This
setting includes degenerate Lagrangian systems with nonholonomic
constraints.

The paper is organized as follows. In Section \ref{Secpre}, we
collect some preliminary notions and geometric objects on Lie
algebroids, including differential calculus, morphism and
prolongations. We also recall the definition and some properties
of (generalized) Dirac structures on vector spaces, vector bundles
and manifolds. In Section \ref{SecILS}, first we introduce and
study the generalized Dirac structure $\D_\U$ on $\T^EE^*$ induced
by a vector subbundle $\U$ of the Lie algebroid $E$. The main goal
of this section is to define implicit Lagrangian systems in terms
of induced Dirac structures. In Section \ref{SecHJ}, we develop a
Hamilton--Jacobi theory for implicit Lagrangian systems on a Lie
algebroid. We apply the results obtained to some particular cases,
in Section \ref{Secex}, recovering some known results. The paper
ends with our conclusions and a description of future research
directions.

\section{Preliminaries}\label{Secpre}

\subsection{Lie algebroids}\label{algebroides} Let $E$ be a vector bundle of rank
$n$ over a manifold $Q$ of dimension $m$ and $\tau:E\to Q$ be the
vector bundle projection. Denote by $\Gamma(E)$ the
$C^\infty(Q)$-module of sections of $\tau:E\to Q$. A \emph{Lie
algebroid structure } $(\lcf\cdot,\cdot\rcf,\rho)$ on $E$ is a Lie
bracket $\lcf\cdot,\cdot\rcf$ on the space $\Gamma(E)$ and a
bundle map $\rho:E\to TQ$, called \emph{the anchor map}, such that
if we also denote by $\rho:\Gamma(E)\to {\mathfrak X}(Q)$ the
homomorphism of $C^\infty(Q)$-modules induced by the anchor map,
then
\[
\lcf X,fY\rcf=f\lcf X,Y\rcf + \rho(X)(f)Y,
\]
for $X,Y\in \Gamma(E)$ and $f\in C^\infty(Q)$. The triple
$(E,\lcf\cdot,\cdot\rcf,\rho)$ is called \emph{a Lie algebroid
over} $Q$ (see \cite{Ma}).

If $(E,\lcf\cdot,\cdot\rcf,\rho)$ is a Lie algebroid over $Q,$
then the anchor map $\rho:\Gamma(E)\to {\mathfrak X}(Q)$ is a
homomorphism between the Lie algebras
$(\Gamma(E),\lcf\cdot,\cdot\rcf)$ and $({\mathfrak
X}(Q),[\cdot,\cdot])$.

Standard examples of Lie algebroids are real Lie algebras of
finite dimension and the tangent bundle $TQ$ of an arbitrary
manifold $Q.$ In more detail, let $(\mathfrak
g,[\cdot,\cdot]_{\mathfrak g})$ be a real Lie algebra of finite
dimension. Then, consider the vector bundle $\tau: \mathfrak g\to
\{\text{ one point }\}$. The section of this vector bundle can be
identified with the elements of $\mathfrak g$ and, therefore, we
can consider the Lie bracket given by the Lie algebra structure
$[\cdot,\cdot]_{\mathfrak g}$ on ${\mathfrak g}$ and the anchor
map $\rho$ given by the null map. So, $(\mathfrak
g,[\cdot,\cdot]_{\mathfrak g},0)$ is a Lie algebroid over a point.
On the other hand, let $Q$ a manifold. The sections of the tangent
bundle $\tau_E=\tau_Q:E=TQ\to Q$ may be identified with the vector
fields on $Q$, the Lie bracket on $\Gamma(\tau_E)={\mathfrak
X}(Q)$ is the usual vector fields bracket and the anchor map is
the identity on $TQ$. Then, the triple $(TQ,[\cdot,\cdot],Id)$ is
a Lie algebroid over $Q$.

Another example of a Lie algebroid may be constructed as follows.
Let $\pi:P\to Q$ be a principal bundle with structure group $G$.
Denote by $\Phi:G\times P\to P$ the free action of $G$ on $P$ and
by $T\Phi:G\times TP\to TP$ the tangent lifted action of $G$ on
$TP$. Then, one may consider the quotient vector bundle
$\tau_P|G:TP/G\to Q=P/G$ and the sections of this vector bundle
may be identified with the vector fields on $P$ which are
invariant under the action $\Phi$. Using the fact that every
$G$-invariant vector field on $P$ is $\pi$-projectable and the
fact that the standard Lie bracket on vector fields is closed with
respect to $G$-invariant vector fields, we can induce a Lie
algebroid structure on $TP/G$. The resultant Lie algebroid is
called \emph{the Atiyah (gauge) algebroid associated with the
principal bundle} $\pi:P\to Q$ (see \cite{LMM,Ma}).

%Finally, another interesting example is the following. Let $(F,
%\lcf\cdot, \cdot\rcf_{F}, \rho_{F})$ be a Lie algebroid over a
%manifold $N$ and $\pi: M \to N$ be a smooth map. \emph{A left
%action of $F$ on $\pi: M \to N$} is a $\R$-linear map
%\[
%\Psi: \Gamma(F) \to {\frak X}(M)
%\]
%such that
%\[
%\begin{array}{l}
%  \Psi (f X) = (f \circ \pi)\Psi (X), \; \; \Psi(\lcf X, Y\rcf_{F}) =
%  -[\Psi(X), \Psi(Y)], \\[5pt] (T_{m}\pi)(\Psi(X)(m)) =
%  -\rho_{F}(X(\pi(m))),
%\end{array}
%\]
%for $f \in C^{\infty}(N)$, $X, Y \in \Gamma(F)$ and $m \in M$. If
%$\Psi: \Gamma(F) \to {\frak X}(M)$ is a left action of $F$ on
%$\pi: M \to N$ and $\tau_{F}: F \to N$ is the vector bundle
%projection then the pullback vector bundle of $F$ over $\pi$,
%\[
%E = F^{*}\pi = \{(f, m) \in F \times M \,|\, \tau_{F}(f) = \pi
%(m)\},
%\]
%is a Lie algebroid over $M$ with Lie algebroid structure
%$(\lcf\cdot, \cdot\rcf_{E}, \rho_{E})$ which is characterized by
%\[
%\lcf X, Y\rcf_{E} = \lcf X, Y\rcf_{F} \circ \pi, \; \; \;
%\rho_{E}(X)(m) = -\Psi(X)(m),
%\]
%for $X, Y \in \Gamma(E)$ and $m \in M$. The triple $(E, \lcf\cdot,
%\cdot\rcf_{E}, \rho_{E})$ is called \emph{the left action Lie
%algebroid
%  of $F$ over $\pi$} and it is denoted by $F_{\Psi}\pi$ (see
%\cite{HiMa}).

Now, let $(E,\lcf\cdot,\cdot\rcf,\rho)$ be a Lie algebroid, then
one may define \emph{the differential} of $E$,
$d^E:\Gamma(\wedge^k E^*)\to \Gamma(\wedge^{k+1}E^*)$, as follows
\begin{multline*} d^E \mu(X_0,\dots, X_k)=
\sum_{i=0}^{k} (-1)^i\rho(X_i)(\mu(X_0,\dots,
\widehat{X_i},\dots, X_k))\\
+\displaystyle\sum_{i< j}(-1)^{i+j}\mu(\lcf X_i,X_j\rcf,X_0,\dots,
\widehat{X_i},\dots,\widehat{X_j},\dots ,X_k),
\end{multline*}
for $\mu\in \Gamma(\wedge^k E^*)$ and $X_0,\dots ,X_k\in
\Gamma(E).$ It follows that $(d^E)^2=0$. Moreover, if
$X\in\Gamma(E)$, one may introduce, in a natural way, \emph{the
Lie derivative with respect to $X$,} as the operator
$\pounds^E_X:\Gamma(\wedge^kE^*)\to \Gamma(\wedge^k E^*)$ given by
$\pounds^E_X=i_X\circ d^E + d^E\circ i_X.$

Note that if $E = TQ$ and $X \in \Gamma(E) = {\mathfrak X}(Q)$
then $d^{TQ}$ and $\pounds_{X}^{TQ}$ are the usual differential
and the usual Lie derivative with respect to $X$, respectively.

If we take local coordinates $(x^i)$ on an open subset $U$ of $Q$
and a local basis $\{e_\alpha\}$ of sections of $E$ defined on
$U$, then we have the corresponding local coordinates
$(x^i,y^\alpha)$ on $E$, where $y^\alpha(e)$ is the $\alpha$-th
coordinate of $e\in E$ in the given basis. Such coordinates
determine local functions $\rho_\alpha^i$, ${\mathcal
C}_{\alpha\beta}^{\gamma}$ on $Q$ which contain local information
about the Lie algebroid structure and, accordingly, they are
called \emph{the structure functions of the Lie algebroid.} They
are given by
\[
\lcf e_\alpha,e_\beta\rcf={\mathcal C}_{\alpha\beta}^\gamma
e_\gamma\quad\text{and}\quad\rho(e_\alpha)=\rho_\alpha^i\frac{\partial
}{\partial x^i}.
\]
These functions should satisfy the relations
\begin{align*}
\rho_\alpha^j\frac{\partial \rho_\beta^i}{\partial x^j}
-\rho_\beta^j\frac{\partial \rho_\alpha^i}{\partial x^j}&=
\rho_\gamma^i{\mathcal C}_{\alpha\beta}^\gamma ,\\
\sum_{cyclic(\alpha,\beta,\gamma)}\Big (
\rho_{\alpha}^i\frac{\partial {\mathcal
C}_{\beta\gamma}^\delta}{\partial x^i} + {\mathcal C}_{\alpha
\nu}^\delta {\mathcal C}_{\beta\gamma}^\nu \Big )&=0,
\end{align*}
which are usually called \emph{the structure equations}.

If $f\in C^\infty(Q)$, we have that
\begin{equation}\label{diff0}
d^E f=\frac{\partial f}{\partial x^i}\rho_\alpha^i e^\alpha,
\end{equation}
where $\{e^\alpha\}$ is the dual basis of $\{e_\alpha\}$. On the
other hand, if $\theta\in \Gamma(E^*)$ and $\theta=\theta_\gamma
e^\gamma$, it follows that
%\begin{equation}\label{diff1}
$$d^E \theta=\Big(\frac{\partial \theta_\gamma}{\partial
x^i}\rho^i_\beta-\frac{1}{2}\theta_\alpha {\mathcal
C}^\alpha_{\beta\gamma}\Big)e^{\beta}\wedge e^\gamma.$$
%\end{equation}
In particular,
\begin{align*}
d^E x^i&=\rho_\alpha^ie^\alpha, \qquad d^E
e^\alpha=-\frac{1}{2}{\mathcal C}_{\beta\gamma}^{\alpha}
e^\beta\wedge e^\gamma.
\end{align*}

\subsection{Morphisms} Let $(E,\lcf\cdot,\cdot \rcf,\rho)$
and $(E',\lcf\cdot,\cdot\rcf', \rho')$ be Lie algebroids over $Q$
and $Q'$, respectively. A morphism of vector bundles $(F,f)$ from
$E$ to $E'$

\begin{picture}(375,90)(80,7)
\put(195,20){\makebox(0,0){$Q$}}
\put(245,25){$f$}\put(210,20){\vector(1,0){80}}
\put(305,20){\makebox(0,0){$Q'$}} \put(185,50){$\tau$}
\put(195,70){\vector(0,-1){40}} \put(310,50){$\tau'$}
\put(305,70){\vector(0,-1){40}} \put(195,80){\makebox(0,0){$E$}}
\put(245,85){$F$}\put(210,80){\vector(1,0){80}}
\put(305,80){\makebox(0,0){$E'$}} \end{picture}

\vspace{0cm} \noindent is a \emph{Lie algebroid morphism} if
\begin{equation}\label{Qorph}
d^E ((F,f)^*\phi')= (F, f)^*(d^{E'}\phi'), \quad \text{for
}\phi'\in \Gamma(\wedge^k(E')^*).
\end{equation}
Note that $(F, f)^*\phi'$ is the section of the vector bundle
$\wedge^kE^*\to Q$ defined by
\[
((F,f)^*\phi')_x(a_1,\dots ,a_k)=\phi'_{f(x)}(F(a_1),\dots
,F(a_k)),
\]
for $x\in Q$ and $a_1,\dots ,a_k\in E_{x}$, where $E_x$ denotes
the fiber of $E$ at the point $x\in Q$. We remark that
(\ref{Qorph}) holds if and only if
%\begin{equation}\label{Qorph1}
\begin{align*}
d^E(g'\circ f)&=(F, f)^{*}(d^{E'}g'),\qquad
\text{for }g'\in C^\infty(Q'),\\
d^E((F,f)^*\alpha')&=(F, f)^*(d^{E'}\alpha'),\qquad\text{for }
\alpha'\in \Gamma((E')^*).
\end{align*}
%\end{equation}

If $(F,f)$ is a Lie algebroid morphism, $f$ is an injective
immersion and $F_{|E_x}:E_x\rightarrow E'_{f(x)}$ is injective,
for all $x\in Q$, then $(E,\lcf\cdot,\cdot\rcf,\rho)$ is said to
be a \emph{Lie subalgebroid} of $(E',\lcf\cdot,\cdot\rcf',\rho')$.

If $Q=Q'$ and $f=id:Q\to Q$ then, it is easy prove that the pair
$(F,id)$ is a Lie algebroid morphism if and only if
\[F\lcf X,Y\rcf=\lcf FX,FY \rcf',\qquad \rho'(FX)=\rho(X),\]
for $X,Y\in \Gamma(E).$

\subsection{Poisson structure on $E^*$} Let $(E,\lcf\cdot,\cdot\rcf,\rho)$ be a Lie algebroid
over $Q$ and $E^*$ be the dual bundle to $E.$ Then, $E^*$ admits a
linear Poisson structure $\Pi_{E^*}$, that is, $\Pi_{E^*}$ is a
$2$-vector on $E^*$ such that
\[
[\Pi_{E^*},\Pi_{E^*}]=0,
\]
and if $f$ and $f'$ are linear functions on $E^*,$ we have that
$\Pi_{E^*}(d^{TE^*}f,d^{TE^*}f')$ is also a linear function on
$E^*$. If $(x^i)$ are local coordinates on $Q$, $\{e_\alpha\}$ is
a local basis of $\Gamma(E)$ and $(x^i,p_\alpha)$ are the
corresponding local coordinates on $E^*$, then the local
expression for $\Pi_{E^*}$ is
%\begin{equation}\label{locPoisson}
$$\Pi_{E^*}=\rho_\alpha^i\frac{\partial }{\partial x^i}\wedge
\frac{\partial }{\partial p_\alpha}-\frac{1}{2}
\C_{\alpha\beta}^\gamma p_\gamma\frac{\partial}{\partial
p_\alpha}\wedge \frac{\partial }{\partial p_\beta},$$
%\end{equation}
where $\rho_\alpha^i$ and $\C_{\alpha\beta}^\gamma$ are the
structure functions of $E$ with respect to the coordinates $(x^i)$
and to the basis $\{e_\alpha\}$. The Poisson structure $\Pi_{E^*}$
induces a linear Poisson bracket of functions on $E^*$ which we
will denote by $\{ \; , \; \}_{E^*}$. In fact, if $F, G \in
C^{\infty}(E^*)$ then \begin{equation}\label{corPoisson} \{F,
G\}_{E^*} = \Pi_{E^*}(d^{TE^*}F, d^{TE^*}G). \end{equation} (For
more details, see \cite{LMM}).

\subsection{The prolongation of a Lie algebroid over a
fibration}\label{secprol}
%In this section we will recall the definition of the Lie %algebroid
%structure on the prolongation of a Lie algebroid over a
%smooth map.

Let $(E,\lcf\cdot,\cdot\rcf,\rho)$ be a Lie algebroid of rank $n$
over a manifold $Q$ of dimension $m$ and $\pi : P \to Q$ be a
fibration, that is, a surjective submersion.

We consider the subset ${\mathcal T}^EP$ of $E\times TP$ defined
by ${\mathcal T}^EP=\displaystyle \bigcup_{p\in P} {\mathcal
T}^E_pP$, where
\[
{\mathcal T}^E_pP=\{(b,v)\in E_{\pi(p)}\times
T_{p}P\,|\,\rho(b)=(T_{p}\pi)(v)\},
\]
and $T\pi:TP\to TQ$ is the tangent map to $\pi$.

Denote by $\tau^{\pi}:{\mathcal T}^EP\to P$ the map given by
\[
\tau^{\pi}(b,v)=\tau_{P}(v),
\]
for $(b,v)\in {\mathcal T}^EP$, where $\tau_{P}:TP\to P$ is the
canonical projection. Then, if $m'$ is the dimension of $P$, one
may prove that
\[
\dim\,{\mathcal T}^E_pP=n+m'-m.
\]
Thus, we conclude that ${\mathcal T}^EP$ is a vector bundle over
$P$ of rank $n+m'-m$ with the vector bundle projection
$\tau^{\pi}:{\mathcal T}^EP\to P.$

A section $\tilde{X}$ of $\tau^{\pi}: {\mathcal T}^{E}P \to P$ is
said to be \emph{projectable} if there exists a section $X$ of
$\tau: E \to Q$ and a vector field $U$ on $P$ which is
$\pi$-projectable to the vector field $\rho(X)$ and such that
$\tilde{X}(p) = (X(\pi(p)), U(p))$, for all $p \in P$. For such a
projectable section $\tilde{X}$, we will use the following
notation $\tilde{X} \equiv (X, U)$. It is easy to prove that one
may choose a local basis of projectable sections of the space
$\Gamma({\mathcal T}^{E}P)$.

The vector bundle $\tau^{\pi}: {\mathcal T}^{E}P \to P$ admits a
Lie algebroid structure $(\lcf \cdot , \cdot \rcf^{\pi},
\rho^{\pi})$. In fact,
\[
\lcf (X_{1}, U_{1}), (X_{2}, U_{2})\rcf^{\pi} = (\lcf X_{1},
X_{2}\rcf, [U_{1}, U_{2}]), \qquad \rho^{\pi}(X_{1}, U_{1}) =
U_{1}.
\]
The Lie algebroid $({\mathcal T}^{E}P, \lcf \cdot , \cdot
\rcf^{\pi}, \rho^{\pi})$ is called \emph{the prolongation of $E$
over $\pi$ or the $E$-tangent bundle to $P$}. Note that if
$\pr_1:{\mathcal T}^EP\to E$ is the canonical projection on the
first factor, then the pair $(\pr_1,\pi)$ is a morphism between
the Lie algebroids $({\mathcal
T}^EP,\lcf\cdot,\cdot\rcf^{\pi},\rho^{\pi})$ and
$(E,\lcf\cdot,\cdot \rcf,\rho)$ (for more details, see
\cite{LMM}).

\begin{example}\label{exTEE}
{\rm Let $(E,\lcf\cdot,\cdot \rcf,\rho)$ be a Lie algebroid of
rank $n$ over a manifold $Q$ of dimension $m$ and $\tau:E\to Q$ be
the vector bundle projection. Consider the prolongation ${\mathcal
T}^EE$ of $E$ over $\tau,$
\[
{\mathcal T}^EE=\{({e},v)\in E\times TE\,|\,
\rho({e})=(T\tau)(v)\}.
\]
${\mathcal T}^EE$ is a Lie algebroid over $E$ of rank $2n$ with
Lie algebroid structure $(\lcf\cdot,\cdot\rcf^{\tau},
\rho^{\tau})$.

If $(x^i)$ are local coordinates on an open subset $U$ of $Q$ and
$\{e_\alpha\}$ is a basis of sections of the vector bundle
$\tau^{-1}(U)\to U$, then $\{{\mathcal X}_\alpha,{\mathcal
V}_\alpha\}$ is a basis of sections of the vector bundle
$(\tau^{\tau})^{-1}(\tau^{-1}(U))\to \tau^{-1}(U)$, where
$\tau^{\tau}:{\mathcal T}^EE\to E$ is the vector bundle projection
and \begin{equation}\label{locbasisTEE} {\mathcal
X}_\alpha(e)=\Big(e_\alpha(\tau(e)),\rho_\alpha^i\displaystyle\frac{\partial
}{\partial x^i}_{|e}\Big),\qquad {\mathcal
V}_\alpha(e)=\Big(0,\displaystyle\frac{\partial }{\partial
y^{\alpha}}_{|e}\Big),
\end{equation}
for $e\in \tau^{-1}(U).$ Here, $\rho_\alpha^i$ are the components
of the anchor map with respect to the basis $\{e_\alpha\}$ and
$(x^i,y^\alpha)$ are the local coordinates on $E$ induced by the
local coordinates $(x^i)$ and the basis $\{e_\alpha\}$. Using the
local basis $\{{\mathcal X}_\alpha,{\mathcal V}_\alpha\}$, one may
introduce, in a natural way, local coordinates
$(x^i,y^\alpha;s^\alpha,w^\alpha)$ on ${\mathcal T}^EE.$ If
$\omega$ is a point of $(\tau^{\tau})^{-1}(\tau^{-1}(U))$, then
$(x^i,y^\alpha)$ are the coordinates of the point
$\tau^{\tau}(\omega)\in \tau^{-1}(U)$ and
\[
\omega=s^\alpha{\mathcal X}_\alpha(\tau^{\tau}(\omega)) + w^\alpha
{\mathcal V}_\alpha(\tau^{\tau}(\omega)).
\]
On the other hand, we have that
\begin{align*}
\lcf{\mathcal X}_\alpha,{\mathcal X}_\beta\rcf^{\tau}&={\mathcal
C}_{\alpha\beta}^\gamma{\mathcal X}_\gamma,& \lcf{\mathcal
X}_\alpha,{\mathcal V}_\beta\rcf^{\tau}&=\lcf {\mathcal
V}_\alpha, {\mathcal V}_\beta\rcf^{\tau}=0,\\
\rho^{\tau}({\mathcal
X}_\alpha)&=\rho_\alpha^i\displaystyle\frac{\partial }{\partial
x^i},& \rho^{\tau}({\mathcal
V}_\alpha)&=\displaystyle\frac{\partial }{\partial y^\alpha},
\end{align*}
for all $\alpha$ and $\beta$, where ${\mathcal
C}_{\alpha\beta}^\gamma$ are the structure functions of the Lie
bracket $\lcf\cdot,\cdot\rcf$ with respect to the basis
$\{e_\alpha\}$.

The vector subbundle $({\mathcal T}^EE)^V$ of ${\mathcal T}^EE$
whose fiber at the point $e\in E$ is
$$({\mathcal T}^E_eE)^V=\{(0,v)\in E\times T_eE\,|\,(T_e\tau)(v)=0\}$$
is called \emph{the vertical subbundle}. Note that $({\mathcal
T}^EE)^V$ is locally generated by the sections $\{{\mathcal
V}_\alpha\}$.

Two canonical objects on ${\mathcal T}^EE$ are \emph{the Euler
section} $\Delta$ and \emph{the vertical endomorphism $S$}.
$\Delta$ is the section of ${\mathcal T}^EE\to E$ locally defined
by
%\begin{equation}
$$\label{Lioulo} \Delta = y^{\alpha}\V_{\alpha},$$
%\end{equation}
and $S$ is the section of the vector bundle $({\mathcal
T}^EE)\otimes ({\mathcal T}^EE)^*\to E$ locally characterized by
the following conditions
\begin{equation}\label{endverlo}
S(\X_{\alpha}) = \V_{\alpha}, \makebox[.3cm]{} S(\V_{\alpha}) = 0,
\makebox[.3cm]{} \text{ for all } \alpha.
\end{equation}
Finally, a section $\xi$ of ${\mathcal T}^EE\to E$ is said to be a
\emph{second-order differential equation} (SODE) on $E$ if
$S(\xi)=\Delta$ or, alternatively, $\pr_1(\xi(e))=e$, for all
$e\in E$ (for more details, see \cite{LMM}). }
\end{example}

\begin{example}
{\rm

Let $(E,\lcf\cdot,\cdot \rcf,\rho)$ be a Lie algebroid of rank $n$
over a manifold $Q$ of dimension $m$ and $\tau^*:E^*\to Q$ be the
vector bundle projection of the dual bundle $E^*$ to $E$.

We consider the prolongation ${\mathcal T}^EE^*$ of $E$ over
$\tau^*,$
\[
{\mathcal T}^EE^*=\{({e}',v)\in E\times TE^*\,|\,
\rho({e}')=(T\tau^*)(v)\}.
\]
${\mathcal T}^EE^*$ is a Lie algebroid over $E^*$ of rank $2n$
with Lie algebroid structure $(\lcf\cdot,\cdot\rcf^{\tau^*},
\rho^{\tau^*})$.

If $(x^i)$ are local coordinates on an open subset $U$ of $Q$,
$\{e_\alpha\}$ is a basis of sections of the vector bundle
$\tau^{-1}(U)\to U$ and $\{e^\alpha\}$ is the dual basis of
$\{e_\alpha\}$, then $\{{\mathcal Y}_\alpha,{\mathcal P}^\alpha\}$
is a basis of sections of the vector bundle
$(\tau^{\tau^*})^{-1}((\tau^*)^{-1}(U))\to (\tau^*)^{-1}(U)$,
where $\tau^{\tau^*}:{\mathcal T}^EE^*\to E^*$ is the vector
bundle projection and
\begin{align}\label{basisTEE*}
{\mathcal
Y}_\alpha(e^*)&=\Big(e_\alpha(\tau^*(e^*)),\rho_\alpha^i\displaystyle\frac{\partial
}{\partial x^i}_{|e^*}\Big),& {\mathcal
P}^\alpha(e^*)=\Big(0,\displaystyle\frac{\partial }{\partial
p_{\alpha}}_{|e^*}\Big),
\end{align}
for $e^*\in (\tau^*)^{-1}(U).$ Here, $(x^i,p_\alpha)$ are the
local coordinates on $E^*$ induced by the local coordinates
$(x^i)$ and the basis $\{e^\alpha\}$ of $\Gamma(E^*)$. Using the
local basis $\{{\mathcal Y}_\alpha,{\mathcal P}^\alpha\}$, one may
introduce, in a natural way, local coordinates
$(x^i,p_\alpha;z^\alpha,u_\alpha)$ on ${\mathcal T}^EE^*.$ If
$\omega^*$ is a point of $(\tau^{\tau^*})^{-1}((\tau^*)^{-1}(U))$,
then $(x^i,p_\alpha)$ are the coordinates of the point
$\tau^{\tau^*}(\omega^*)\in (\tau^*)^{-1}(U)$ and
\[
\omega^*=z^\alpha{\mathcal Y}_\alpha(\tau^{\tau^*}(\omega^*)) +
u_\alpha {\mathcal P}^\alpha(\tau^{\tau^*}(\omega^*)).
\]
On the other hand, we have that
\begin{equation}
\label{rhoTEE*}
\begin{aligned}
\lcf{\mathcal Y}_\alpha,{\mathcal Y}_\beta\rcf^{\tau^*}&={\mathcal
C}_{\alpha\beta}^\gamma{\mathcal Y}_\gamma, &\qquad \lcf{\mathcal
Y}_\alpha,{\mathcal P}^\beta\rcf^{\tau^*}&=\lcf {\mathcal
P}^\alpha, {\mathcal
P}^\beta\rcf^{\tau^*}=0,\\
\rho^{\tau^*}({\mathcal
Y}_\alpha)&=\rho_\alpha^i\displaystyle\frac{\partial }{\partial
x^i},& \rho^{\tau^*}({\mathcal
P}^\alpha)&=\displaystyle\frac{\partial }{\partial p_\alpha},
\end{aligned}
\end{equation}

\noindent for all $\alpha$ and $\beta$. Thus, if $\{{\mathcal
Y}^\alpha,{\mathcal P}_\alpha\}$ is the dual basis of $\{{\mathcal
Y}_\alpha,{\mathcal P}^\alpha\}$, then
%\begin{equation}\label{dif*}
\begin{align*}
d^{{\mathcal T}^EE^*}f&= \rho_\alpha^i\displaystyle\frac{\partial
f}{\partial x^i}{\mathcal Y}^\alpha+ \displaystyle\frac{\partial
f}{\partial p_\alpha} {\mathcal
P}_\alpha,\\
d^{{\mathcal T}^EE^*}{\mathcal Y}^\gamma&= \displaystyle
-\frac{1}{2} {\mathcal C}_{\alpha\beta}^{\gamma} {\mathcal
Y}^\alpha \wedge {\mathcal Y}^\beta,\\
d^{{\mathcal T}^EE^*}{\mathcal P}_\gamma&=0,
\end{align*}
%\end{equation}

\noindent for $f\in C^\infty(E^*).$

We may introduce a canonical section $\lambda_E$ of the vector
bundle $({\mathcal T}^EE^*)^*\to E^*$ as follows. If $e^*\in E^*$
and $(\tilde{e},v)$ is a point of the fiber of ${\mathcal T}^EE^*$
over $e^*$, then
%\begin{equation}\label{Lio}
$$\lambda_E(e^*)(\tilde{e},v)=\langle e^*,\tilde{e}\rangle,$$
% \end{equation}
where $\langle\cdot,\cdot\rangle$ denotes the natural pairing
between $E^*$ and $E$. $\lambda_E$ is called \emph{the Liouville
section} of $({\mathcal T}^EE^*)^*$.

Now, \emph{the canonical symplectic section} $\Omega_E$ is the
nondegenerate closed 2-section defined by
%\begin{equation}\label{sym}
$$\Omega_E=-d^{{\mathcal T}^EE^*} \lambda_E.$$
%\end{equation}
Then, we have that the map $\Omega_E^\flat:\T^EE^*\to (\T^EE^*)^*$
defined as
\begin{equation}\label{omegaflat}
\Omega_E^\flat(X)=i_X\Omega_E,
\end{equation}
for all $X\in \T^EE^*$, where $i_X$ denote the contraction by $X$,
is a vector bundles isomorphism.

In local coordinates,
%\begin{equation}\label{locLio}
$$\lambda_E(x^i,p_\alpha)=p_\alpha{\mathcal Y}^\alpha ,$$
%\end{equation}
\begin{equation}\label{locsym}
\Omega_E(x^i,p_\alpha)={\mathcal Y}^\alpha\wedge {\mathcal
P}_\alpha + \frac{1}{2} {\mathcal C}_{\alpha\beta}^\gamma p_\gamma
{\mathcal Y}^\alpha\wedge {\mathcal Y}^\beta .
\end{equation}

\begin{remark} The linear Poisson bracket $\{\cdot,\cdot\}_{E^*}$
on $E^*$ induced by the Lie algebroid structure on $E$ (see
(\ref{corPoisson})) can be also defined in terms of the canonical
symplectic 2-section $\Omega_E$. In fact, for $F, G \in
C^\infty(E^*)$, we have that
$$\{F,G\}_{E^*}=\Omega_E((\Omega_E^\flat)^{-1}(d^{\T^EE^*}F),(\Omega_E^\flat)^{-1}(d^{\T^EE^*}G)).$$

\end{remark}

}
\end{example}

\subsection{Dirac structures}

In this section we briefly recall the definition and some
properties of Dirac structures on vector spaces, vector bundles
and manifolds (see \cite{C,CW}). The construction of a Dirac
structure will be reviewed, which will be important for defining
implicit Lagrangian systems.

Let $V$ be an $n$-dimensional vector space, $V^*$ be its dual
space, and let $\langle\cdot,\cdot\rangle$ be the natural pairing
between $V^*$ and $V$. A \emph{Dirac structure} on $V$ is a
subspace $\D\subset V\oplus V^*$ such that $\D=\D^\perp$, where
$\D^\perp$ is the orthogonal complement of $\D$, that is,
$$\D^\perp=\{(u,\beta)\in V\oplus V^* \,|\,
\langle\beta,v\rangle+\langle\alpha,u\rangle=0, \text{ for all
}(v,\alpha)\in\D\}.$$

It is easy to prove that a vector subspace $\D\subset V\oplus V^*$
is a Dirac structure on $V$ if and only if $\dim \D=n$ and
$\langle\alpha,\bar{v}\rangle+\langle\bar{\alpha},v\rangle=0$, for
all $(v,\alpha),(\bar{v},\bar{\alpha})\in \D$. From the definition
of a Dirac structure, for each $(v,\alpha)\in\D$, we have that
$\langle\alpha,v\rangle=0$.

If $V$ is a vector bundle over a manifold $Q$, let $V\oplus_Q V^*$
be the Whitney sum bundle over $Q$, that is, it is the bundle over
the base $Q$ and with fiber over the point $x\in Q$ equal to
$V_x\times V^*_x$, where $V_x$ (respectively, $V^*_x$) is the
fiber of $V$ (respectively, $V^*$) at the point $x$. A \emph{Dirac
structure} on $V$ is a subbundle $\D\subset V\oplus_Q V^*$ that is
a Dirac structure in the sense of vector spaces at each point
$x\in Q$.

Now, let $M$ be a smooth differentiable manifold and $\tau_M:TM\to
M$ its tangent bundle. An \emph{almost} (in the terminology of
\cite{YM1}) \emph{or generalized} (in the terminology of
\cite{DVS}) \emph{Dirac structure} on $M$ is a subbundle
$\D\subset TM\oplus_M T^*M$ which is a Dirac structure in the
sense of vector bundles.

In geometric mechanics, almost Dirac structures provide a
simultaneous generalization of both 2-forms (not necessarily
closed and possibly degenerate) as well as almost Poisson
structures (that is, bracket that need not satisfy the Jacobi
identity). A \emph{Dirac structure} on $M$ is an almost Dirac
structure that additionally satisfies the following integrability
condition
$$\langle\pounds_{X_1}\alpha_2,X_3\rangle+\langle\pounds_{X_2}\alpha_3,X_1\rangle+\langle\pounds_{X_3}\alpha_1,X_2\rangle=0,$$
for all $(X_1,\alpha_1),(X_2,\alpha_2),(X_3,\alpha_3)\in\D$, and
where $\pounds_X$ denotes the usual Lie derivative with respect to
the vector field $X$. This generalizes closedness for the
symplectic form, and the Jacobi identity for Poisson structures.
For the remainder of this paper, we will primarily be concerned
with almost Dirac structures, since it allows one to incorporate
nonholonomic constraints.

Two constructions of almost Dirac structures on a manifold are
given as follows. The first construction is induced by a
distribution and a 2-form on the manifold. Let $M$ be a manifold,
$\Omega$ be a 2-form on $M$ and $\Delta_M$ be a distribution on
$M$. Denote by $\Omega^\flat$ the associated flat map and by
$\Delta_M^\circ\subset T^*M$ the annihilator of $\Delta_M$. Then,
from Theorem 2.3 in \cite{YM1}, we have that $\D_M\subset
TM\oplus_M T^*M$ defined, for each $x\in M$, by
\begin{align*}
\D_M(x)=&\{(v_x,\alpha_x)\in T_xM\times T^*_xM \,|\,
v_x\in\Delta_M(x) \text{ and }\\
&\qquad\qquad \alpha_x-\Omega^\flat(x)(v_x)\in\Delta_M^\circ(x)\}
\end{align*}
is an almost Dirac structure on $M$ (see also Theorem 3.2 in
\cite{DVS}).

For the case when $M=T^*Q$ and $\Omega=\Omega_{T^*Q}$ is the
canonical symplectic 2-form, this almost Dirac structure was used
to introduce the notion of implicit Lagrangian systems in standard
mechanics (see \cite{YM1,YM2}).

The second almost Dirac structure is induced by a codistribution
and a skew-symmetric 2-tensor on the manifold. Let $M$ a manifold,
$\Pi:T^*M\times T^*M\to\R$ be a skew-symmetric 2-tensor and
$\Delta^*_M$ a codistribution on $M$. Denote by
$\sharp_{\Pi}:T^*M\to TM$ the associated sharp map and by $\Ker
\Delta^*_M$ the distribution on $M$ defined as
$$\Ker \Delta^*_M=\{ X\in TM \,|\, \alpha(X)=0, \text{ for all
}\alpha\in\Delta^*_M\}.$$ Then, from Theorem 2.4 in \cite{YM1}, we
have that $\D_M\subset TM\oplus_M T^*M$ defined, for each $x\in
M$, by
\begin{align*}
\D_M(x)=&\{(v_x,\alpha_x)\in T_xM\times T^*_xM \,|\,
\alpha_x\in\Delta_M^*(x) \text{ and }\\
&\qquad\qquad v_x-\sharp_\Pi(x)(\alpha_x)\in\Ker
\Delta_M^\circ(x)\}
\end{align*}
is an almost Dirac structure on $M$ (see also Theorem 3.1 in
\cite{DVS}).

For the case when $M=T^*Q$ and $\Pi=\Pi_{T^*Q}$ is the canonical
Poisson structure on $T^*Q$, this almost Dirac structure coincides
with the almost Dirac structure described before which was used to
introduce the notion of implicit Lagrangian systems in standard
mechanics (see \cite{YM1,YM2}).

\section{Implicit Lagrangian systems on a Lie
algebroid}\label{SecILS}

\subsection{Induced almost Dirac structure}
First, we introduce the notion of an \emph{induced almost Dirac
structure} on the Lie algebroid prolongation $\T^EE^*$ of a Lie
algebroid $\tau:E\to Q$. This almost Dirac structure is induced by
a vector subbundle $\U$ of $E$, that is, $\U\subset E$ such that
$\tau_\U=\tau_{|\U}:\U\to Q$ is a vector bundle.

Consider the dual vector bundle $\tau^*:E^*\to Q$ of $\vb$. We can
define its prolongation to the corresponding prolongation Lie
algebroids $\T\tau^*:\T^EE^*\to \T^EQ$ as the identity in the
first component and the tangent map of $\tau^*$ in the second,
that is, $\T\tau^*=(id, T\tau^*)$. It is easy to prove that it is
a Lie algebroid morphism between $\T^EE^*\to E^*$ and $\T^EQ\to Q$
(see \cite{M1} for a general definition of the prolongation of a
map). Moreover, we can identify $\T^EQ$ with $E$ and then
$\T\tau^*\equiv \pr_1$, $\pr_1:\T^EE^*\to E$ being the projection
on the first factor.

The vector subbundle $\U$ can be lifted to a vector subbundle
$\U_{\T^EE^*}\subset \T^EE^*$ as follows
\begin{equation}\label{defUTEE}
\U_{\T^EE^*}=(\pr_1)^{-1}(\U).
\end{equation}
Denote by $\U^\circ_{\T^EE^*}\subset (\T^EE^*)^*$ its annihilator.
Then, we have the following result.

\begin{theorem}\label{thDU} Let $\alg$ be a Lie algebroid over a manifold $Q$ and
$\U$ be a vector subbundle of $E$. For each $e^*\in E^*$, let
\begin{equation}\label{defDU}
\begin{aligned}
\D_\U(e^*)=&\{(X_{e^*},\alpha_{e^*})\in\T^E_{e^*}E^*\times
(\T^E_{e^*}E^*)^* \,|\, X_{e^*}\in \U_{\T^EE^*}(e^*) \text{ and }\\
&\qquad\qquad \alpha_{e^*}-\Omega_E^\flat(e^*)(X_{e^*})\in
\U_{\T^EE^*}^\circ(e^*)\}.
\end{aligned}
\end{equation}
Then, $\D_\U\subset \T^EE^*\oplus_{E^*}(\T^EE^*)^*$ is an almost
Dirac structure on $\T^EE^*$.
\end{theorem}
\begin{proof}

First, it is not difficult to prove that, since $\U_{\T^EE^*}$ is
a vector subbundle of $\T^EE^*$, $\D_{\U}$ is a vector subbundle
of $\T^EE^*\oplus_{E^*} (\T^EE^*)^*$.

Second, the orthogonal of $\D_{\U}\subset \T^EE^*\oplus_{E^*}
(\T^EE^*)^*$ is given at $e^*\in E^*$ by
\begin{align*}
\D_\U^\perp(e^*)=&\{(Y_{e^*},\beta_{e^*})\in\T^E_{e^*}E^*\times
(\T^E_{e^*}E^*)^* \,|\,
\alpha_{e^*}(Y_{e^*})+\beta_{e^*}(X_{e^*})=0, \\
&\qquad\qquad \text{ for all } X_{e^*}\in \U_{\T^EE^*}(e^*) \text{
and } \alpha_{e^*}-\Omega_E^\flat(e^*)(X_{e^*})\in
\U_{\T^EE^*}^\circ(e^*)\}.
\end{align*}

To check that $\D_{\U}(e^*)\subset \D_{\U}^\perp (e^*)$, we
consider $(X_{e^*},\alpha_{e^*})\in \D_{\U}(e^*)$ and then, for
any $(X_{e^*}',\alpha_{e^*}')\in \D_{\U}(e^*)$, we have that
$$
\alpha_{e^*}(X'_{e^*})+\alpha'_{e^*}(X_{e^*})=\Omega_E(e^*)(X_{e^*},X_{e^*}')+\Omega_E(e^*)(X_{e^*}',X_{e^*})=0,$$
by the skew-symmetry of $\Omega_E$. This implies that
$(X_{e^*},\alpha_{e^*})\in \D_{\U}^\perp(e^*)$. Therefore,
\begin{equation}\label{subset1}
\D_{\U}(e^*)\subset \D_{\U}^\perp (e^*).
\end{equation}

Now, to prove that $\D_{\U}^\perp(e^*)\subset \D_{\U}(e^*)$, let
$(Y_{e^*},\beta_{e^*})\in \D_{\U}^\perp (e^*)$. Then, we have that
\begin{equation}\label{condperp}
\alpha_{e^*}(Y_{e^*})+\beta_{e^*}(X_{e^*})=0, \end{equation}  for
all $(X_{e^*},\alpha_{e^*})\in\T^E_{e^*}E^*\times
(\T^E_{e^*}E^*)^*$ such that $X_{e^*}\in \U_{\T^EE^*}(e^*)$ and
$\alpha_{e^*}-\Omega_E^\flat(e^*)(X_{e^*})\in
\U_{\T^EE^*}^\circ(e^*)$. If we choose $X_{e^*}=0$ and
$\alpha_{e^*}\in \U_{\T^EE^*}^\circ (e^*)$, then
$(X_{e^*},\alpha_{e^*})\in\D_{\U}(e^*)$. Therefore, using
(\ref{condperp}), we obtain that $\alpha_{e^*}(Y_{e^*})=0$, for
all $\alpha_{e^*}\in \U_{\T^EE^*}^\circ (e^*)$. Then, we conclude
$Y_{e^*}\in\U_{\T^EE^*}(e^*)$. On the other hand, let
$X_{e^*}\in\U_{\T^EE^*}(e^*)$ be arbitrary and suppose that
$\alpha_{e^*}(Z_{e^*})=\Omega_E(e^*)(X_{e^*},Z_{e^*})$, for all
$Z_{e^*}\in\U_{\T^EE^*}(e^*)$. Since
$Y_{e^*}\in\U_{\T^EE^*}(e^*)$, we have
$\alpha_{e^*}(Y_{e^*})=\Omega_E(e^*)(X_{e^*},Y_{e^*})$ and, from
(\ref{condperp}), we deduce that
$$
\Omega_E(e^*)(X_{e^*},Y_{e^*})+\beta_{e^*}(X_{e^*})=0,
$$ for all $X_{e^*}\in\U_{\T^EE^*}(e^*)$. This implies that
$\beta_{e^*}-\Omega_E^\flat(Y_{e^*})\in\U_{\T^EE^*}^\circ(e^*)$.
Therefore, $(Y_{e^*},\beta_{e^*})\in \D_{\U}(e^*)$ and thus
\begin{equation}\label{subset2}
\D_{\U}^\perp(e^*)\subset \D_{\U}(e^*).
\end{equation}

Given (\ref{subset1}) and (\ref{subset2}), we conclude that
$\D_{\U}^\perp(e^*)= \D_{\U}(e^*)$, and the result follows.

\end{proof}

In what follows, we will obtain a local representation of the
almost Dirac structure $\D_{\U}$ induced on $\T^E E^*$ by a vector
subbundle $\U$ of $E$. Consider local coordinates $(x^i)$ on $Q$,
a local basis $\{e_\alpha\}$ of sections of $E$ and the
corresponding local coordinates $(x^i,y^\alpha)$ on $E$. Let
$\{\Y_\alpha,\P^\alpha\}$ be the local basis of
$\tau^{\tau^*}:\T^EE^*\to E^*$ defined by (\ref{basisTEE*})
induced by the local coordinates $(x^i)$ on $Q$ and the local
basis $\{e_\alpha\}$ of $E$ and $(x^i,p_\alpha;z^\alpha,u_\alpha)$
be the induced local coordinates on $\T^EE^*$.

Thus, we can locally represent the fiber of $\U_{\T^EE^*}$ at a
point $(x^i,p_\alpha)\in E^*$ as
$$\U_{\T^EE^*}(x^i,p_\alpha)=\{(x^i,p_\alpha;z^\alpha,u_\alpha)
\,|\, (x^i,z^\alpha)\in\U(x^i)\}.$$ If we denote by
$(x^i,p_\alpha;r_\alpha,v^\alpha)$ the corresponding local
coordinates induced on $(\T^EE^*)^*$ by the dual basis
$\{\Y^\alpha,\P_\alpha\}$ of $\{\Y_\alpha,\P^\alpha\}$, then the
annihilator of $\U_{\T^EE^*}$ is locally given by
$$\U_{\T^EE^*}^\circ(x^i,p_\alpha)=\{(x^i,p_\alpha;r_\alpha,v^\alpha)
\,|\, v^\alpha=0 \text{ and } (x^i,r_\alpha)\in\U^\circ(x^i) \}.$$

From (\ref{locsym}), we have that
\begin{equation}\label{localOmegaflat}
\Omega_E^\flat(x^i,p_\alpha)(x^i,p_\alpha;z^\alpha,u_\alpha)=(x^i,p_\alpha;-u_\alpha
-\C_{\alpha\beta}^\gamma p_\gamma z^\beta,z^\alpha)
\end{equation}
and then the condition
$\alpha_{e^*}-\Omega_E^\flat(e^*)(X_{e^*})\in
\U^\circ_{\T^EE^*}(e^*)$ can be written locally as
$$v^\alpha=z^\alpha\quad\text{and}\quad (x^i,r_\alpha+u_\alpha+\C_{\alpha\beta}^\gamma p_\gamma
z^\beta)\in\U^\circ(x^i),$$ where
$X_{e^*}\equiv(x^i,p_\alpha;z^\alpha,u_\alpha)$ and
$\alpha_{e^*}\equiv(x^i,p_\alpha;r_\alpha,v^\alpha)$.

Finally, we obtain that
\begin{equation}\label{locDU}
\begin{aligned}
\D_\U(e^*)=&\{(X_{e^*},\alpha_{e^*})\in\T^E_{e^*}E^*\times
(\T^E_{e^*}E^*)^* \,|\, X_{e^*}\in \U_{\T^EE^*}(e^*) \text{ and }\\
&\qquad\qquad \alpha_{e^*}-\Omega_E^\flat(e^*)(X_{e^*})\in
\U_{\T^EE^*}^\circ(e^*)\}\\
 =&\{ ((x^i,p_\alpha;z^\alpha,u_\alpha),(x^i,p_\alpha;r_\alpha,v^\alpha)) \,|\,
 (x^i,z^\alpha)\in\U(x^i),\\
 &\qquad\qquad v^\alpha=z^\alpha \text{ and } (x^i,r_\alpha+u_\alpha+\C_{\alpha\beta}^\gamma p_\gamma
z^\beta)\in\U^\circ(x^i)\} .
\end{aligned}
\end{equation}

\begin{remark}\label{rmcoor} One of the advantages of working in the Lie algebroids
setting is that we can construct a local basis $\{e_\alpha\}$ of
sections of $E$ as follows. We take a local basis $\{e_a\}$ of
sections of the vector bundle $\tau_{\U}:\U\to Q$ and complete it
to a basis $\{e_a,e_A\}$ of local sections of $E$. In this way, we
have coordinates $(x^i,y^\alpha)=(x^i,y^a,y^A)$ on $E$. In this
set of coordinates, the equations which define the subbundle $\U$
are $y^A=0$. So, we can consider $(x^i,y^a)$ as local coordinates
on $\U$. Moreover, if $\{e^a,e^A\}$ is the dual basis of
$\{e_a,e_A\}$ of $E^*$, then $\{e^A\}$ is a local basis of
sections of $\U^\circ$. So, from the definition of $\U_{\T^EE^*}$,
we deduce that $\{\Y_a,\P^a,\P^A\}$ is a local basis of sections
of $\U_{\T^EE^*}\to E^*$ and, if $\{\Y^a,\Y^A,\P_a,\P_A\}$ is the
dual basis of $\{\Y_a,\Y_A,\P^a,\P^A\}$, then $\{\Y^A\}$ is a
local basis of $\U_{\T^EE^*}^\circ$. Therefore, a local
representation for the almost Dirac structure $\D_\U$ is
\begin{align*}
\D_\U(x^i,p_\alpha)=&\{
((x^i,p_\alpha;z^\alpha,u_\alpha),(x^i,p_\alpha;r_\alpha,v^\alpha))
\,|\, z^A=v^A=0, \; v^a=z^a,\\
 &\qquad\qquad \text{ and } r_a=-u_a-\C_{ab}^\gamma p_\gamma
z^b\} .
\end{align*}
\end{remark}

We have used the canonical symplectic section $\Omega_E$ on
$\T^EE^*$ together with a vector subbundle $\U\subset E$ to define
the almost Dirac structure $\D_\U$. However there is a dual
version of the above construction in which the almost Dirac
structure is defined by a Poisson structure on $\T^EE^*$ together
with a vector subbundle $\U\subset E$.

Let $\alg$ be a Lie algebroid and $\U$ be a vector subbundle of
$E$. Consider the projection $\pi^2:(\T^EE^*)^*\to E$ defined as
$\pi^2=\pr_1\comp (\Omega_E^\flat)^{-1}$, where $\pr_1:\T^EE^*\to
E$ is the projection on the first factor. If we consider local
coordinates as before, using (\ref{localOmegaflat}), we have that
$$\pi^2(x^i,p_\alpha;r_\alpha,v^\alpha)=(x^i,v^\alpha).$$
Now, we define the induced vector subbundle $\U^*_{\T^EE^*}$ of
$(\T^EE^*)^*$ by
$$\U^*_{\T^EE^*}=(\pi^2)^{-1}(\U).$$
Note that $\U^*_{\T^EE^*}=\Omega_E^\flat(\U_{\T^EE^*})$, from the
definition of $\U_{\T^EE^*}$ (see (\ref{defUTEE})). Locally,
$\U_{\T^EE^*}^*$ is given by
$$\U_{\T^EE^*}^*(x^i,p_\alpha)=\{(x^i,p_\alpha;r_\alpha,v^\alpha)
\,|\, (x^i,v^\alpha)\in\U(x^i)\}.$$ The annihilator of
$\U^*_{\T^EE^*}$ is given, for each $e^*\in E^*$, by
\begin{align*}
(\U^*_{\T^EE^*})^\circ(e^*)=&\{ X_{e^*}\in \T^E_{e^*}E^* \,|\,
\alpha_{e^*}(X_{e^*})=0, \text{ for all
}\alpha_{e^*}\in\U^*_{\T^EE^*}(e^*)\}\\
=&\{ X_{(x^i,p_\alpha)}=(x^i,p_\alpha;z^\alpha,u_\alpha) \,|\,
z^\alpha=0 \text{ and }(x^i,u_\alpha)\in\U^\circ(x^i)\}.
\end{align*}

On the other hand, we introduce the section $\Pi$ of the vector
bundle $\wedge^2(\T^EE^*)\to E^*$ defined by
%\begin{equation}\label{Pi}
$$\Pi(\alpha,\beta)=\Omega_E((\Omega_E^\flat)^{-1}(\alpha),(\Omega_E^\flat)^{-1}(\beta)),$$
%\end{equation}
for $\alpha,\beta\in (\T^EE^*)^*$. $\Pi$ is the algebraic Poisson
structure on the vector bundle $\T^EE^*\to E^*$ associated with
the symplectic section $\Omega_E$. Denote by
$\sharp_\Pi:(\T^EE^*)^*\to\T^EE^*$ the vector bundles morphism
given by
$$\sharp_\Pi(\alpha)=-i_\alpha\Pi, \text{ for
}\alpha\in(\T^EE^*)^*.$$ Note that
$\sharp_\Pi=(\Omega_E^\flat)^{-1}$.

Then, using the above notation, the induced almost Dirac structure
$\D_\U$ on $\T^EE^*$ is given, for $e^*\in E^*$, by
\begin{align*}
\D_\U(e^*)=&\{(X_{e^*},\alpha_{e^*})\in\T^E_{e^*}E^*\times
(\T^E_{e^*}E^*)^* \,|\, \alpha_{e^*}\in \U_{\T^EE^*}^*(e^*) \text{
and
}\\
&\qquad\qquad X_{e^*}-\sharp_\Pi(e^*)(\alpha_{e^*})\in
(\U_{\T^EE^*}^*)^\circ(e^*)\},
\end{align*}
whose local representation is
\begin{align*}
\D_\U(x^i,p_\alpha)=&\{((x^i,p_\alpha;z^\alpha,u_\alpha),(x^i,p_\alpha;r_\alpha,v^\alpha))
\,|\, (x^i,v^\alpha)\in \U(x^i),\\
&\qquad\qquad v^\alpha=z^\alpha \text{ and } (x^i,
r_\alpha+u_\alpha+\C_{\alpha\beta}^\gamma p_\gamma z^\beta)\in
\U^\circ(x^i)\},
\end{align*}
which coincides with (\ref{locDU}).

\subsection{Implicit Lagrangian systems on a Lie
algebroid}\label{secILS} In this section, an implicit Lagrangian
system on a Lie algebroid $E$ is defined in the context of the
induced almost Dirac structure $\D_\U$ on $\T^EE^*$. As we shall
see, the notion of implicit Lagrangian systems that is developed
here can handle systems with degenerate Lagrangians as well as
systems with nonholonomic constraints. Another description to
address these systems was recently presented by Grabowska and
Grabowski in \cite{GG}, where they use the notion of a Lie
algebroid as a double vector bundle morphism.

Let $L:E\to\R$ be a Lagrangian function on the Lie algebroid
$\alg$.

First of all, we will recall the definition of the Legendre
transformation in the context of Lie algebroids. Given a
Lagrangian function $L:E\to \R$, one can consider the
Poincar\'{e}-Cartan $1$-section associated with $L$, $\theta_L\in
\Gamma((\T^EE)^*)$, which is given
%\begin{equation}\label{cartan1}
$$\theta_L(e)(Z_e)=(d^{\T^EE}
L(e))(S_e(Z_e))=\rho^\tau(S_e(Z_e))(L),$$
%\end{equation}
for $e\in E$ and $Z_e\in \T^E_eE$, $S:\T^EE\to \T^EE$ being the
vertical endomorphism defined in (\ref{endverlo}). So, \emph{the
Legendre transformation associated with $L$} is defined as the
smooth map $\FL:E\to E^*$ defined by
%\begin{equation}\label{LegL}
$$\FL(e)(e')=\theta_L(e)(Z),$$
%\end{equation}
for $e,e'\in E,$ where $Z\in \T^E_eE$ such that $\pr_1(Z)=e'$,
$\pr_1:\T^EE\to E$ being the canonical projection over the first
factor. For more details see \cite{LMM}.

The map $\FL$ is well-defined and its local expression in fiber
coordinates on $E$ and  $E^*$ is
%\begin{equation}\label{locLegL}
$$\FL(x^i,y^\alpha)=\Big(x^i,\frac{\partial L}{\partial
y^\alpha}\Big).$$
%\end{equation}

Now, we consider the isomorphism $A_E:\T^EE^*\to(\T^EE)^*$ between
the vector bundles $\pr_1:\T^EE^*\to E$ and
$(\tau^\tau)^*:(\T^EE)^*\to E$ introduced in \cite{LMM} and whose
local expression is
\begin{equation}\label{localAE}
A_E(x^i,p_\alpha;z^\alpha,u_\alpha)=(x^i,z^\alpha;u_\alpha+\C_{\alpha\beta}^\gamma
p_\gamma z^\beta,p_\alpha). \end{equation}

Then, we define the map $\gamma_E:(\T^EE)^*\to (\T^EE^*)^*$ as
$\gamma_E=\Omega_E^\flat\comp A_E^{-1}$ which is an isomorphism
between the vector bundles $(\tau^\tau)^*:(\T^EE)^*\to E$ and
$\pr_1^*:(\T^EE^*)^*\to E$. From (\ref{localOmegaflat}) and
(\ref{localAE}), we deduce that the local expression of this
isomorphism is
\begin{equation}\label{locgammaE}
\gamma_E(x^i,y^\alpha;s_\alpha,w_\alpha)=(x^i,w_\alpha;-s_\alpha,y^\alpha).
\end{equation}

Now, define a differential operator \textbf{D} acting on the
Lagrangian $L:E\to \R$, which we shall call the \emph{Dirac
differential} of $L$ by
$$\textbf{D}L: E\to (\T^EE^*)^*, \quad \textbf{D}L=\gamma_E\comp
d^{\T^EE}L,$$ where $d^{\T^EE}L$ is the differential of $L$ on the
Lie algebroid $\T^EE$ which is a section of
$(\tau^{\tau})^*:(\T^EE)^*\to E$.

Using (\ref{diff0}), (\ref{locbasisTEE}) and (\ref{locgammaE}), we
conclude that $\textbf{D}L$ is represented in local coordinates by
\begin{equation}\label{locDL}
\textbf{D}L(x^i,y^\alpha)=\Big(x^i,\displaystyle\frac{\partial
L}{\partial y^\alpha}; -\rho_\alpha^i\frac{\partial L}{\partial
x^i},y^\alpha\Big).
\end{equation}

Now, we have all the ingredients to define an implicit Lagrangian
system on a Lie algebroid.

\begin{definition} Let $L:E\to\R$ be a given Lagrangian function (possibly
degenerate) on a Lie algebroid $\alg$ and $\U\subset E$ be a given
vector subbundle of $\tau:E\to Q$. Denote by $\D_\U$ the induced
almost Dirac structure on the Lie algebroid prolongation $\T^EE^*$
that is given by (\ref{defDU}) and $\textbf{D}L:E\to(\T^EE^*)^*$
the Dirac differential of $L$. Let $P=\FL(\U)\subset E^*$ be the
image of $\U$ under the Legendre transformation.

An {\bf implicit Lagrangian system} is a triple $(L,\U,X)$, where
$X$ is a section of the Lie algebroid prolongation
$\tau^{\tau^*}:\T^EE^*\to E^*$ defined at the points of $P$,
together with the condition
%\begin{equation}\label{IL}
$$(X,\textbf{D}L)\in\D_\U.$$
%\end{equation}
In other words, as $P=\FL(\U)\subset E^*$, $X$ can be seen as a
section of $\T^EE^*\to E\oplus_Q E^*$ defined at the points of
$\U\oplus_Q P$ and thus, we require that for each point $e\in\U$
and with $e^*=\FL(e)\in P$, we have
$$(X(e,{e^*}),\textbf{D}L(e))\in\D_\U(e^*).$$

\end{definition}

\begin{definition} A {\bf solution curve} of an implicit
Lagrangian system $(L,\U,X)$ is a curve $(x(t),y(t))\in\U(x(t))$,
$t_1\leq t \leq t_2$, such that $\FL(x(t),y(t))$ is an integral
curve of the vector field $\rho^{\tau^*}(X)$ on $E^*$,
$\rho^{\tau^*}$ being the anchor map of the Lie algebroid
$\tau^{\tau^*}:\T^EE^*\to E^*$.
\end{definition}

\begin{remark} One can consider the map $i_E:E\to E\oplus_Q E^*$
defined as the direct sum of the identity map on $E$, $id:E\to E$,
and the Legendre transformation $\FL:E\to E^*$. Denote by
${\mathcal K}$ the submanifold of $E\oplus_Q E^*$ defined as the
image of $\U$ under $i_E$. Thus, ${\mathcal K}$ is locally given
by
$${\mathcal K}=\{(x^i,y^\alpha,p_\alpha)\in E_x\times E_x^* \,|\,
(x^i,y^\alpha)\in \U(x^i), \quad
p_\alpha=\displaystyle\frac{\partial L}{\partial y^\alpha}\}.$$

Another way to define the submanifold ${\mathcal K}$ is the
following. Consider the map $\rho_{(\T^EE)^*}:(\T^EE)^*\to
E\oplus_{Q}E^*$ defined as the direct sum of the maps
$(\tau^{\tau})^*:(\T^EE)^*\to E$ and $\tau^{\tau^*}\comp
(A_E)^{-1}:(\T^EE)^*\to E^*$. Recall that
$(\tau^{\tau})^*:(\T^EE)^*\to E$ is the projection of the dual
vector bundle of the Lie algebroid prolongation of $E$ over the
fibration $\tau$, $A_E:\T^EE^*\to (\T^EE)^*$ is the vector bundle
isomorphism defined in (\ref{localAE}) and
$\tau^{\tau^*}:\T^EE^*\to E$ is the projection of the Lie
algebroid prolongation of $E$ over $\tau^*$. If we consider local
coordinates introduced in Section \ref{secprol}, the map
$\rho_{(\T^EE)^*}$ is given by
\begin{equation}\label{locrhoTEE*}
\rho_{(\T^EE)^*}(x^i,y^\alpha,s_\alpha,w_\alpha)=(x^i,y^\alpha,w_\alpha).
\end{equation}
Note that when $E$ is the standard Lie algebroid, that is, $E=TQ$,
then this map is the map $\rho_{T^*TQ}:T^*TQ\to TQ\oplus_Q T^*Q$
defined in \cite{YM1} (see Section 4.10 in \cite{YM1}).

Then, we can construct the map ${i}_E$ between $E$ and $E\oplus_Q
E^*$ by the composition
$${i}_E=\rho_{(\T^EE)^*}\comp
A_E\comp(\Omega_E^\flat)^{-1}\comp\textbf{D}L:E\to
(\T^EE^*)^*\to\T^EE^*\to (\T^EE)^*\to E\oplus_Q E^*,$$
$\textbf{D}L$ being the Dirac differential of the Lagrangian
function $L$ and $\Omega_E^\flat$ being the flat map defined by
the canonical symplectic section $\Omega_E$ (see
(\ref{omegaflat})).

From (\ref{localOmegaflat}), (\ref{localAE}), (\ref{locDL}) and
(\ref{locrhoTEE*}), we have that the local expression of $i_E$ is
$$i_E(x^i,y^\alpha)=\Big(x^i,y^\alpha,\displaystyle\frac{\partial
L}{\partial y^\alpha}\Big).$$

Then, the submanifold ${\mathcal K}\subset E\oplus_Q E^*$ can be
defined as ${\mathcal K}=i_E(\U).$

Then, a solution of an implicit Lagrangian system $(L,\U,X)$ may
be equivalently defined to be a curve $(x(t),y(t),p(t))$, where
$t_1\leq t \leq t_2$, whose image lies in the submanifold
${\mathcal K}\subset E\oplus_Q E^*$ and such that $(x(t),p(t))$ is
an integral curve of $\rho^{\tau^*}(X)$ and such that
$$(X(x(t),y(t),p(t)),\textbf{D}L(x(t),y(t)))\in\D_\U(x(t),p(t)).$$
\end{remark}

Locally, using the preceding notation, (\ref{rhoTEE*}),
(\ref{locDU}) and (\ref{locDL}), we deduce that a solution curve
$(x^i(t),$ $y^\alpha(t),p_\alpha(t))$ for an implicit Lagrangian
system $(L,\U,X)$ must satisfy the following equations
\begin{equation}\label{locIL}
\begin{cases}
(x^i,y^\alpha)\in\U(x^i), \quad \dot{x}^i=\rho_\alpha^i y^\alpha,\quad p_\alpha=\displaystyle\frac{\partial L}{\partial y^\alpha},\\
\Big(x^i, \dot{p}_\alpha+\C_{\alpha\beta}^\gamma p_\gamma y^\beta
-\rho_\alpha^i\displaystyle\frac{\partial L}{\partial
x^i}\Big)\in\U^\circ(x^i).
\end{cases}
\end{equation}

\begin{remark} If we consider the local coordinates on $E$ introduced in Remark
\ref{rmcoor}, the implicit Lagrangian equations reduce to
$$
y^A=0, \quad \dot{x}^i=\rho_a^i y^a, \quad
p_\alpha=\displaystyle\frac{\partial L}{\partial y^\alpha}
\quad\text{and}\quad \dot{p}_a=-\C_{ab}^\gamma
\displaystyle\frac{\partial L}{\partial y^\gamma}
y^b+\rho_a^i\frac{\partial L}{\partial x^i}.$$
\end{remark}

\subsection{Conservation of energy}
Define the \emph{generalized energy} $E_L:E\oplus_Q E^*\to \R$ by
$$E_L(x,e, e^*)=\langle e^*,e\rangle-L(x,e),$$
where $(x,e)\in\U$ and $(x,e^*)\in P$.

\begin{proposition} Let $(x(t),y(t))$, $t_1\leq t\leq t_2$, be an
integral curve of a given implicit Lagrangian system $(L,\U,X)$ on
a Lie algebroid $E$. Then, the function $E_L(x(t),y(t),p(t))$ is
constant in time, where $p(t)=$ $\frac{\partial L}{\partial
y}(x(t),y(t))$.
\end{proposition}
\begin{proof} We give the proof using local coordinates. Then, from the definition of the generalized energy $E_L$,
we have that
$$\displaystyle\frac{d E_L}{dt}= \dot{y}^\alpha p_\alpha+ y^\alpha \dot{p}_\alpha
 -\frac{\partial L}{\partial x^i}
\dot{x}^i-\frac{\partial L}{\partial y^\alpha}\dot{y}^\alpha.$$ As
$p_\alpha(t)=\displaystyle\frac{\partial L}{\partial
y^\alpha}(x^i(t),y^\beta(t))$, we deduce that
\begin{align*}
\displaystyle\frac{d E_L}{dt}&= y^\alpha\dot{p}_\alpha
-\displaystyle\frac{\partial L}{\partial x^i} \dot{x}^i\\
&=y^\alpha\Big(
\dot{p}_\alpha+\C_{\alpha\beta}^\gamma\displaystyle\frac{\partial
L}{\partial y^\gamma} y^\beta-\rho_\alpha^i\frac{\partial
L}{\partial
x^i}\Big)-y^\alpha\C_{\alpha\beta}^\gamma\displaystyle\frac{\partial
L}{\partial y^\gamma} y^\beta+y^\alpha\rho_\alpha^i\frac{\partial
L}{\partial x^i}-\frac{\partial L}{\partial x^i} \dot{x}^i.
\end{align*}

Now, as $(x^i(t),y^\alpha(t))$ satisfies the implicit Lagrangian
equations (\ref{locIL}), we know that
$$(x^i,y^\alpha)\in\U(x^i),\quad \dot{x}^i=\rho_\alpha^i
y^\alpha \quad\text{and}\quad
\Big(x^i,\dot{p}_\alpha+\C_{\alpha\beta}^\gamma\displaystyle\frac{\partial
L}{\partial y^\gamma} y^\beta-\rho_\alpha^i\frac{\partial
L}{\partial x^i}\Big)\in\U^\circ(x^i).$$ Moreover, as
$\C_{\alpha\beta}^\gamma=-\C_{\beta\alpha}^\gamma$, the term
$y^\alpha\C_{\alpha\beta}^\gamma\displaystyle\frac{\partial
L}{\partial y^\gamma} y^\beta=0$. So, we conclude that
$$\displaystyle\frac{dE_L}{dt}=0.$$
\end{proof}

\section{Hamilton--Jacobi theory for implicit Lagrangian
systems}\label{SecHJ}

Let $\alg$ be a Lie algebroid over a manifold $Q$ with projection
$\tau:E\to Q$ and $(L,\U,X)$ be an implicit Lagrangian system on
$E$.

\begin{theorem}\label{HJth} Let $\tilde{\gamma}:Q\to E\oplus_Q E^*$ be a section of
the canonical projection $\nu:E\oplus_Q E^*\to Q$ such that
\begin{equation}\label{gammacond1}
\tilde{\gamma}(Q)\subset {\mathcal K},
\end{equation}
and
\begin{equation}\label{gammacond2}
d^E(\pr_{E^*}\comp\tilde{\gamma})_{|\U\times\U}=0.
\end{equation}
Denote by $\sigma\in\Gamma(E)$ the section $\sigma=\pr_1\comp
X\comp\pr_{E^*}\comp \tilde{\gamma}$, where $\pr_1:\T^EE^*\to E$
is the projection on the first factor and $\pr_{E^*}:E\oplus_Q
E^*\to E^*$ the projection over the second component. Then, the
following conditions are equivalent:
\begin{enumerate}
\item For every curve $c:I\to Q$ in $Q$ such that
\begin{equation}\label{curvecond}
\dot{c}(t)=\rho(\sigma)(c(t)), \,\text{ for all }\,t,
\end{equation}
the curve $\tilde{\gamma}\comp c$ is a solution of the implicit
Lagrangian system $(L,\U,X)$.
\item $\tilde{\gamma}$ satisfies the
\textbf{Hamilton--Jacobi equation for implicit Lagrangian
systems}:
\begin{equation}\label{HJeq}
d^E(E_L\comp\tilde{\gamma})\in\U^\circ.
\end{equation}
\end{enumerate}
\end{theorem}

\begin{proof} We give the proof using local coordinates. We
consider local coordinates $(x^i)$ on an open subset $V$ of $Q$
and a local basis $\{e_\alpha\}$ of sections of $E$ defined on
$V$, then we have the corresponding local coordinates
$(x^i,y^\alpha)$ on $E$. Denote by $\rho_\alpha^i$ and ${\mathcal
C}_{\alpha\beta}^{\gamma}$ the structure functions of the Lie
algebroid $E$ with respect to $(x^i)$ and $\{e_\alpha\}$.

Suppose that
$\tilde{\gamma}(x^i)=(x^i,\gamma^\alpha(x^j),\bar{\gamma}_\alpha(x^j))$.
Then, the condition (\ref{gammacond1}) means that
\begin{equation}\label{gammacond1loc}
(x^i,\gamma^\alpha(x^j))\in\U(x^i) \quad\text{and}\quad
\bar{\gamma}_\alpha(x^i)=\displaystyle\frac{\partial L}{\partial
y^\alpha}(x^i,\gamma^\alpha(x^j)),
\end{equation} and the condition
(\ref{gammacond2}) can be written locally as
\begin{equation}\label{gammacond2loc}
\displaystyle\frac{\partial \bar{\gamma}_\delta}{\partial x^i}
\rho^i_\beta v^\beta w^\delta =\Big(\frac{\partial
\bar{\gamma}_\beta}{\partial x^i}\rho^i_\delta+\bar{\gamma}_\alpha
\C_{\beta\delta}^\alpha\Big) v^\beta w^\delta,
\end{equation}
for all $v, w \in\U$ given locally by $v=v^\beta e_\beta$ and
$w=w^\delta e_\delta$.

If $c(t)=(c^i(t))$, it is easy to prove that equation
(\ref{curvecond}) can be rewritten in local coordinates as
\begin{equation}\label{curvecondloc}
\dot{c}^i(t)=\gamma^\alpha(c(t))\rho_\alpha^i(c(t)).
\end{equation}

Using the hypothesis (\ref{gammacond1}) (see its local expression
(\ref{gammacond1loc})), we also have that equation (\ref{HJeq}) is
locally written as
\begin{equation}\label{HJeqloc}
\Big(\gamma^\beta\displaystyle\frac{\partial\bar{\gamma}_\beta}{\partial
x^i}-\frac{\partial L}{\partial x^i}\Big)\rho^i_\alpha v^\alpha=0,
\end{equation}
for all $v=v^\alpha e_\alpha\in \U$.

(i) $\Rightarrow$ (ii) Assume that (i) holds. Therefore
\begin{equation}\label{condi}
\begin{cases}
 (c^i(t),\gamma^\alpha(c(t)))\in\U(c(t)),\\
\dot{c}^i(t)=\gamma^\alpha(c(t))\rho^i_\alpha(c(t)), \\
\bar{\gamma}_\alpha(c(t))=\displaystyle\frac{\partial
L}{\partial y^\alpha}(c(t),\gamma^\beta(c(t))),\\
\Big(\displaystyle\frac{\partial\bar{\gamma}_\alpha}{\partial
x^j}(c(t))\dot{c}^j(t)+\C_{\alpha\beta}^\delta(c(t))
\bar{\gamma}_\delta(c(t))\gamma^\beta(c(t))\\
\qquad\qquad -\rho^j_\alpha(c(t)) \displaystyle\frac{\partial
L}{\partial
x^j}(c(t),\gamma^\beta(c(t)))\Big)e^\alpha(c(t))\in\U^\circ(c(t)).
\end{cases}
\end{equation}

Then, using (\ref{gammacond2loc}) and (\ref{condi}), we have that
at $c(t)$, for $w=w^\alpha e_\alpha(c(t))\in\U(c(t))$,
\begin{align*}
0&=\Big(\displaystyle\frac{\partial\bar{\gamma}_\alpha}{\partial
x^j} \gamma^\beta\rho_\beta^j+\C_{\alpha\beta}^\delta
\bar{\gamma}_\delta\gamma^\beta-\rho^j_\alpha
\displaystyle\frac{\partial L}{\partial
x^j} \Big) w^\alpha\\
&=\Big( \displaystyle\frac{\partial\bar{\gamma}_\beta}{\partial
x^j}
\gamma^\beta\rho_\alpha^j+\bar{\gamma}_\delta\C_{\beta\alpha}^\delta
\gamma^\beta +\C_{\alpha\beta}^\delta
\bar{\gamma}_\delta\gamma^\beta-\rho^j_\alpha
\displaystyle\frac{\partial L}{\partial
x^j}\Big) w^\alpha\\
&=\Big( \displaystyle\frac{\partial\bar{\gamma}_\beta}{\partial
x^j} \gamma^\beta- \displaystyle\frac{\partial L}{\partial x^j}
\Big)\rho^j_\alpha w^\alpha.
\end{align*}
Then, we conclude that (ii) holds (see (\ref{HJeqloc})).

(ii) $\Rightarrow$ (i) Suppose that (ii) holds, that is, condition
(\ref{HJeqloc}) is satisfied. Let $c:I\to Q$ a curve such that
$\dot{c}(t)=\rho(\sigma)(c(t))$. Then, we have that
%\begin{equation}\label{iipart1}
$$\dot{c}^i(t)=\gamma^\alpha(c(t))\rho_\alpha^i(c(t)).$$
%\end{equation}
Moreover, from (\ref{gammacond1loc}), we also know that
\begin{equation}\label{iipart2}
(c^i(t),\gamma^\alpha(c(t)))\in\U(c(t)) \quad \text{and}\quad
\bar{\gamma}_\alpha(c(t))=\displaystyle\frac{\partial L}{\partial
y^\alpha}(c(t),\gamma^\alpha(c(t))),
\end{equation}

Moreover, using (\ref{gammacond1loc}) and (\ref{gammacond2loc}),
we deduce that at $c(t)$, for all $w=w^\alpha
e_\alpha(c(t))\in\U(c(t))$,
\begin{equation}\label{iipart3}
\begin{aligned}
\quad\,\Big(\displaystyle\frac{\partial\bar{\gamma}_\alpha}{\partial
x^j}\dot{c}^j&+\C_{\alpha\beta}^\delta
\bar{\gamma}_\delta\gamma^\beta-\rho^j_\alpha
\displaystyle\frac{\partial L}{\partial
x^j}\Big) w^\alpha\\
&=\Big( \displaystyle\frac{\partial\bar{\gamma}_\beta}{\partial
x^j}
\gamma^\beta\rho_\alpha^j+\bar{\gamma}_\delta\C_{\beta\alpha}^\delta
\gamma^\beta +\C_{\alpha\beta}^\delta
\bar{\gamma}_\delta\gamma^\beta -\rho^j_\alpha
\displaystyle\frac{\partial L}{\partial
x^j} \Big) w^\alpha\\
&=\Big( \displaystyle\frac{\partial\bar{\gamma}_\beta}{\partial
x^j} \gamma^\beta- \displaystyle\frac{\partial L}{\partial x^j}
\Big)\rho^j_\alpha w^\alpha=0.
\end{aligned}
\end{equation}

Finally, from (\ref{iipart2}) and (\ref{iipart3}), we conclude
that $\tilde{\gamma}\comp c$ is an integral curve of ${X}$.

\end{proof}

\section{Examples}\label{Secex}

\subsection{The case $\U=E$}

This is perhaps the simplest case in which one has no constraints
but the Lagrangian may be degenerate. In this case, the induced
almost Dirac structure $\D_\U=\D_E$ is given by
$$\D_E(e^*)=\{(X_{e^*},\alpha_{e^*})\in
\T^E_{e^*}E^*\times(\T^E_{e^*}E^*)^* \,|\,
\alpha_{e^*}=\Omega_E^\flat(e^*)(X_{e^*})\}.$$ So, locally the
equations defining the almost Dirac structure in this case are
$$v^\alpha=z^\alpha\quad \text{and}\quad r_\alpha=-u_\alpha-\C_{\alpha\beta}^\gamma p_\gamma z^\beta,$$
where $X_{e^*}\equiv (x^i,p_\alpha;z^\alpha,u_\alpha)$ and
$\alpha_{e^*}\equiv (x^i,p_\alpha;r_\alpha,v^\alpha)$. Then, a
curve $(x^i(t),y^\alpha(t))$ in $E$ is a solution of the implicit
Lagrangian system if and only if
$$p_\alpha=\displaystyle\frac{\partial L}{\partial y^\alpha}, \;\dot{x}^i=\rho_\alpha^i y^\alpha \quad\text{and}\quad
\dot{p}_\alpha=\rho_\alpha^i\frac{\partial L}{\partial
x^i}-\C_{\alpha\beta}^\gamma p_\gamma y^\beta.$$ This means that
in this case, the condition of an implicit Lagrangian system is
equivalent to the Euler--Lagrange equations for $L$ (see Equations
(2.40) in \cite{LMM}).

In the usual formulation of Lagrangian systems on a Lie algebroid,
one must to restrict to admissible curves on the Lie algebroid
$E$, that is, curves $c(t)$ in $E$ such that $(c(t),\dot{c}(t))\in
\T^E_{c(t)}E$ or, locally, if $c(t)=(x^i(t),y^\alpha(t))$ then
$\dot{x}^i=\rho_\alpha^i y^\alpha$. Notice that integral curves of
an implicit Lagrangian system automatically satisfy this
condition.

In this case we can reformulate the Theorem \ref{HJth} as follows.

\begin{corollary}Let ${\gamma}:Q\to E$ be a section of
the Lie algebroid $\tau:E\to Q$ such that
$$d^E(\FL\comp{\gamma})=0.$$
Denote by $\sigma\in\Gamma(E)$ the section $\sigma=\pr_1\comp
X\comp \FL\comp{\gamma}$, where $\pr_1:\T^EE^*\to E$ is the
projection on the first factor. Then, the following conditions are
equivalent:
\begin{enumerate}
\item For every curve $c:I\to Q$ in $Q$ such that
$$\dot{c}(t)=\rho(\sigma)(c(t)), \,\text{ for all }\,t,$$
the curve ${\gamma}\comp c$ is a solution of the implicit
Lagrangian system $(L,E,X)$, or equivalently, $\gamma\comp c$ is a
solution of the Euler--Lagrange equations for the Lagrangian $L$.
\item ${\gamma}$ satisfies
$$d^E(\varepsilon_L\comp {\gamma})=0,$$
where $\varepsilon_L:E\to \R$ is the energy function associated
with $L$ (see (2.39) in \cite{LMM}) which is defined as
$\varepsilon_L=\rho^\tau(\Delta)(L)-L$, $\Delta\in\Gamma(\T^EE)$
is the Euler section and $\rho^\tau$ is the anchor map of the Lie
algebroid $\tau^\tau:\T^EE\to E$ (see Example \ref{exTEE}).
\end{enumerate}
\end{corollary}

This result can be viewed as the Lagrangian version of the Theorem
3.16 in \cite{LMM}.

\subsection{The case $E=TQ$}

Let $E$ be the standard Lie algebroid $\tau_{TQ}:TQ\to Q$. In this
case, the sections of this vector bundle can be identified with
vector fields on $Q$, the Lie bracket of sections is just the
usual Lie bracket of vector fields and the anchor map is the
identity map $id:TQ\to TQ$. A vector subbundle $\U$ of $TQ$ is
just a distribution $\Delta_Q$ on $Q$.

Moreover, in this case, the Lie algebroid
$(\T^EE^*,\lcf\cdot,\cdot\rcf^{\tau^*},\rho^{\tau^*})$ is the
standard Lie algebroid $(TT^*Q,[\cdot,\cdot],id)$. So, the lift of
the vector subbundle $\U=\Delta_Q$ to $\T^EE^*=TT^*Q$ is just the
distribution $\Delta_{T^*Q}$ on $T^*Q$ defined as
$$\Delta_{T^*Q}=(T\pi_Q)^{-1}(\Delta_Q),$$
$\pi_Q:T^*Q\to Q$ being the canonical projection. Moreover,
$\Omega_E=\Omega_{TQ}$ is the canonical symplectic 2-form on
$T^*Q$. Then, the induced almost Dirac structure $\D_{\U}$ defined
in Theorem \ref{thDU} is given, for each point $z\in T^*Q$, by
\begin{align*}
\D_{\Delta_Q}=&\{ (v_z,\alpha_z)\in T_zT^*Q\times T^*_zT^*Q |
v_z\in \Delta_{T^*Q}(z) \text{ and }\\
&\qquad\qquad \alpha_z-\Omega_{TQ}^\flat(z)(v_z)\in
\Delta_{T^*Q}^\circ(z)\}. \end{align*}

This almost Dirac structure coincides with the induced almost
Dirac structure introduced by Yoshimura and Marsden in \cite{YM1}.
Thus, if we apply the results of Section \ref{secILS} to this
particular case we recover the formulation of implicit Lagrangian
systems develop in \cite{YM1}.

Moreover, applying the Theorem \ref{HJth} to this particular case
one recover the result develop in \cite{LOS} for standard implicit
Lagrangian systems.

\begin{example}\label{exEPred}{\rm We are going to consider a simple example: the
case of Euler--Poincar\'e reduction. In this case, we consider the
particular case when the manifold $Q$ is a Lie group $G$ and the
distribution $\Delta_Q$ is just $TQ=TG$ (that is, the cases when
$\U=E$ and $E=TG$). Let $L:TG\to\R$ be a left-invariant
Lagrangian. Then, we have that $(g(t),v(t))\in T_{g(t)}G$,
$t_1\leq t\leq t_2$, is a solution curve of the implicit
Lagrangian system $(L,TG,X)$ if and only if $g(t)$ is a solution
of the Euler--Lagrange equations for $L$ on $G$ and
$\dot{g}(t)=v(t)$, for $t_1\leq t\leq t_2$.

On the other hand, let ${\mathfrak g}$ be the Lie algebra
associated with $G$ which is a Lie algebroid over a point. As, $L$
is a left-invariant function, we can consider the reduced
Lagrangian $l:{\mathfrak g}\to\R$, $l=L_{|{\mathfrak g}}$. Taking
$\U={\mathfrak g}$, we have that a curve $\xi(t)\in {\mathfrak g}$
is a solution of the implicit Lagrangian system $(l,{\mathfrak
g},Y)$ if and only if it is a solution of the Euler--Poincar\'e
equations on ${\mathfrak g}$. Moreover, as well known, $g(t)$
satisfies the Euler--Lagrange equations for $L$ on $G$ if and only
if $\xi(t)={g(t)}^{-1}\dot{g}(t)$ satisfies the Euler--Poincar\'e
equations on ${\mathfrak g}$. Then, we conclude that $(g(t),v(t))$
is a solution curve of the implicit Lagrangian system $(L,TG,X)$
if and only if $\xi(t)={g(t)}^{-1}\dot{g}(t)$ is a solution curve
of the implicit Lagrangian system $(l,{\mathfrak g},Y)$ and
$\dot{g}(t)=v(t)$.

}
\end{example}

\subsection{Nonholonomic mechanics on Lie
algebroids}

Let $\alg$ be a Lie algebroid. Nonholonomic constraints on the Lie
algebroids setting are given by a vector subbundle $\U$ of $E$. In
\cite{CoLeMaMa}, the authors introduced the notion of a
nonholonomically constrained Lagrangian system on a Lie algebroid
$E$ as a pair $(L,\U)$, where $L:E\to \R$ is a Lagrangian function
on $E$ and $\U$ is the constraint subbundle, that is, it is a
vector subbundle of $E$.

If we consider local coordinates as in Remark \ref{rmcoor}, then a
solution curve $(x^i(t),$ $y^\alpha(t))$, $t_1\leq t\leq t_2$, on
$E$ for a nonholonomic system must satisfy the differential
equations (see Equations (3.7) in \cite{CoLeMaMa})
\[
\begin{cases}
\dot{x}^i=\rho_a^i y^a,\\[6pt]
\displaystyle\frac{d}{dt}\Big( \frac{\partial L}{\partial y^a}
\Big)+\frac{\partial L}{\partial y^\gamma}\C_{ab}^\gamma
y^b-\rho_a^i \frac{\partial L}{\partial x^i}=0,\\[8pt]
y^A=0.
\end{cases}
\]

So, a nonholonomic system is locally represented by an implicit
Lagrangian system $(L,\U,X)$ together with the condition
$$p_\alpha(t)=\displaystyle\frac{\partial L}{\partial y^\alpha}(x^i(t),y^\alpha(t)),$$
since $(x^i(t),p_\alpha(t))=\FL(x^i(t),y^\alpha(t))$, where $\FL$
is the Legendre transformation.

\begin{example}{\rm
Consider the situation of Example \ref{exEPred} but with a
non-trivial left-invariant distribution on $G$. Then, we have
$Q=G$ a Lie group, $L:TG\to\R$ a left-invariant Lagrangian and
$\U=\Delta_G$, $\Delta_G\neq TG$ and $\Delta_G\neq \{0\}$, a
left-invariant distribution on $G$, that is, a standard
nonholonomic LL system on $G$. As we have proved in general,
$(g(t),v(t),p(t))$, $t_1\leq t\leq t_2$, is a solution curve of
the implicit Lagrangian system $(L,TG,X)$ if and only if
$(g(t),v(t))$ is a solution of the Lagrange-d'Alembert equations
for $L$ on $G$ and $p(t)=\partial L/\partial v(g(t),v(t))$, for
$t_1\leq t\leq t_2$.

On the other hand, this type of nonholonomic system on $G$ may be
reduced to a nonholonomic system on the Lie algebra ${\mathfrak
g}$ associated with $G$. The reduced Lagrangian $l:{\mathfrak
g}\to\R$ is $l=L_{|{\mathfrak g}}$ and the vector subspace
${\mathfrak d}$ of ${\mathfrak g}$ is given by ${\mathfrak
d}=\Delta_G(e)$. Then, one has a constrained system $(l,{\mathfrak
d})$ on ${\mathfrak g}$. So, a curve $\xi(t)\in {\mathfrak g}$ is
a solution of the implicit Lagrangian system $(l,{\mathfrak g},Y)$
if and only if it is a solution of the constrained
Euler--Poincar\'e equations (or the so-called
Euler--Poincar\'{e}--Suslov equations, see \cite{FeZe}) on ${\mathfrak
g}$.

}
\end{example}

\section{Conclusion and Future work}\label{Secfut}

In this paper, we introduced the notion of an induced almost Dirac
structure, and show how it leads to implicit Lagrangian systems on
Lie algebroids. This provides a generalization of Lagrangian
mechanics on Lie algebroids that can address degenerate
Lagrangians as well as holonomic and nonholonomic constraints.
Furthermore, we have obtained a Hamilton--Jacobi theory for such
systems.

In future research, we aim to study the possibility of obtaining a
Hamilton--Jacobi equation, as in the general case described in
\cite{LMM2}, using the notion of Dirac algebroids given in
\cite{GG}. In this case, the theory will include all important
cases of Lagrangian and Hamiltonian systems, including systems
with and without constraints, and autonomous and non-autonomous
systems.

Another interesting direction would be to generalize to Dirac algebroids the work in \cite{LOS} that relates the Hamilton--Jacobi theory of a Dirac mechanical system with symmetry and the Hamilton--Jacobi theory of the associated reduced Dirac system. Furthermore, the relationship between the various Hamilton--Jacobi theories for reduction of Dirac mechanical systems formulated on Dirac, Courant, and Lie algebroids, and the formulations based on Lagrange--Poincar\'e bundles~\cite{YM3} also remains to be studied.

\section*{Acknowledgements}
This material is based upon work supported by the National Science
Foundation under the applied mathematics grant DMS-0726263, the
Faculty Early Career Development (CAREER) award DMS-1010687, the
FRG award DMS-1065972, MICINN (Spain) grants MTM2009-13383 and
MTM2009-08166-E, and the projects of the Canary government
SOLSUBC200801000238 and ProID20100210. We also want to thank the two anonymous referees for their helpful comments and suggestions.

\end{document}